\documentclass[11pt,letterpaper]{article}
\usepackage[margin=1in]{geometry}
 \usepackage{authblk}
\usepackage{graphicx,amssymb,amsmath,xcolor}
\usepackage[disable]{todonotes}  

\newtheorem{theorem}{Theorem}[section]

\newtheorem{lemma}[theorem]{Lemma}
\newtheorem{remark}[theorem]{Remark}

\newenvironment{proof}{\textbf{Proof:}}{\hfill$\square$}

\author[$\ddagger$]{J. Mark Keil}
\author[$\ddagger$]{Fraser McLeod}
\author[$\ddagger$]{Debajyoti Mondal\thanks{Corresponding author: Debajyoti Mondal (d.mondal@usask.ca)}}
\affil[$\ddagger$]{Department of Computer Science, University of Saskatchewan, Saskatoon,  Canada}  

\date{}

\usepackage{graphicx} 
\usepackage{amssymb}
\usepackage{amsmath}
\usepackage{scrextend}
\usepackage{hanging}
\usepackage{xcolor}

\usepackage[noend]{algpseudocode}
\usepackage{algorithm}
\algrenewcommand\algorithmicindent{2em}%

\usepackage{graphicx}
\graphicspath{ {./images/} }

\title{Quantum Speedup for Some Geometric 3SUM-Hard \\
Problems and Beyond\thanks{The work is supported in part by the Natural Sciences and Engineering Research Council of Canada (NSERC).}} 

\begin{document}
\thispagestyle{empty}
\maketitle

\begin{abstract}
The classical 3SUM conjecture states that the class of 3SUM-hard problems does not admit a truly subquadratic $O(n^{2-\delta})$-time algorithm, where $\delta >0$, in classical computing. The geometric 3SUM-hard problems have widely been studied in computational geometry and recently, these problems have been examined under the quantum computing model. For example, Ambainis and Larka [TQC'20] designed a quantum algorithm that can solve many geometric 3SUM-hard problems in $O(n^{1+o(1)})$-time, whereas Buhrman [ITCS'22] investigated lower bounds under quantum 3SUM conjecture that claims there does not exist any sublinear $O(n^{1-\delta})$-time quantum algorithm for the 3SUM problem. The main idea of Ambainis and Larka is to formulate a 3SUM-hard problem as a search problem, where one needs to find a point with a certain property over a set of regions determined by a line arrangement in the plane. The quantum speed-up then comes from the application of the well-known quantum search technique called Grover search over all regions.  
 
This paper further generalizes the technique of Ambainis and Larka for some 3SUM-hard problems when a solution may not necessarily correspond to a single point or the search regions do not immediately correspond to the subdivision determined by a line arrangement. Given a set of $n$ points and a positive number $q$, we design $O(n^{1+o(1)})$-time quantum algorithms to determine whether there exists a triangle among these points with an area at most $q$ or a unit disk that contains at least $q$ points. We also give an $O(n^{1+o(1)})$-time quantum algorithm to determine whether a given set of intervals can be translated so that it becomes contained in another set of given intervals and discuss further generalizations. 

 






\end{abstract}

\section{Introduction} 
A rich body of research investigates ways to speed up algorithmic computations by using quantum computing techniques.  Grover's algorithm~\cite{DBLP:conf/stoc/Grover98} (a quantum search algorithm) has often been leveraged to obtain quadratic speedup for various problems over the classical solution.
For example, consider the problem of finding a specific item within an unordered database of $n$  items. In the classical setting, this task requires $\Omega(n)$ operations. However,  with high probability, Grover's algorithm can find the item in \(O(\sqrt{n})\) quantum operations~\cite{grover1996fast}. 
In this paper we investigate quantum speedup for some geometric  
3SUM-hard problems. 

Given a set $S$ of $n$ numbers, the 3SUM problem asks whether there are elements $a,b,c \in S$ such that $a+b+c = 0$. The class of 3SUM-hard problems consists of problems that are at least as hard as the 3SUM problem. The classical 3SUM conjecture states that the class of 3SUM-hard problems cannot be solved in $O(n^{2-\delta})$ time, where $\delta >0$, in a classical computer~\cite{gajentaan1995class}. However,   3SUM can be solved in $O(n \log n)$ time in a quantum computer by applying Grover search over all possible pairs as follows~\cite{ambainis2020quantum}: We have $O(n)$ quantum search operations to resolve and if we maintain the elements of $S$ in a balanced binary search tree, then for each pair $a,b$, we can decide the existence of $-(a+b)\in S$ in $O(\log n)$ time. In general, such straightforward quantum speedup does not readily apply to all problems even if they can be solved in $O(n^2)$ time in a classic computer \cite[Table 1]{DBLP:conf/innovations/BuhrmanLPS22}.

Ambainis and Larka~\cite{ambainis2020quantum}  designed a quantum algorithm that can solve many geometric 3SUM-hard problems in $O(n^{1+o(1)})$-time. Some examples are Point-On-3-Lines, Triangles-Cover-Triangle, and Point-Covering. The Point-On-3-Lines problem takes a set of lines as input and asks to determine whether there is a point that lies on at least 3 lines.  The Triangles-Cover-Triangle problem asks whether a given set of triangles in the plane covers another given triangle. Given a set of $n$ half-planes and an integer $t$, the Point-Covering problem asks whether there is a point that hits at least $t$ half-planes. 

The idea of Ambainis and Larka~\cite{ambainis2020quantum} is to model these problems as a point search problem over a subdivision of the plane with a small number of regions. 
Specifically, consider a random set of $k$ lines in the Point-On-3-Lines problem and a triangulation of an arrangement of these lines, which subdivides the plane into $O(k^2)$ regions. We can check each corner of these regions to check whether it hits at least three lines in $O(nk^2)$ time. Otherwise, 
we can search each region recursively by taking only the lines that intersect the region into consideration. It is known that with high probability every subproblem size (i.e., the number of lines intersecting a region) would be small~\cite{DBLP:journals/dcg/Clarkson87, DBLP:conf/compgeom/HausslerW86}, and one can obtain a running time of $O(n^{1+o(1)})$ by a careful choice of $k$ and by the application of Grover search~\cite{ambainis2020quantum}.  For the Point-Covering problem, one can construct a similar subdivision of the plane using $k$ random half-planes. We can then count for each region $R$, the number $i$ of half-planes fully covering $R$ in $O(nk^2)$ time, and if a solution is not found, then recursively search in the subproblem for a point that hits at least $(t-i)$ half-planes. For the Triangles-Cover-Triangle problem, one can construct the $O(k^2)$-size subdivision (of the given triangle $T$ which we want to cover) by $k$ lines determined by $k$ segments that are randomly chosen from the boundaries of the set $S$ of given triangles. In $O(nk^2)$ time we can determine the regions of $T$ that are fully covered by a single triangle of $S$. For every remaining region $R$, let $S(R)$ be a set of triangles where each intersects $R$ but does not fully contain $R$. We now can search over all such regions $R$ recursively for a point that is not covered by $S(R)$. 

In this paper we show how Ambainis and Larka's~\cite{ambainis2020quantum} idea 
can be adapted even for problems where a solution may not correspond to a single point or the search regions do not necessarily correspond to a subdivision determined by an arrangement of straight lines. Specifically, we show that the following problems admit an $O(n^{1+o(1)})$-time quantum algorithm.

    \smallskip\noindent\textsc{$q$-Area Triangle}: Given a set $S$ of $n$ points, decide whether they determine a triangle with area at most $q$.
    
    \smallskip\noindent\textsc{$q$-Points in a Disk}: Given a set $S$ of $n$ points, determine whether there is a unit disk that covers at least $q$ of these points. 
    
    \smallskip\noindent\textsc{Interval Containment}: Given two sets $P$ and $Q$ of pairwise-disjoint intervals on a line, where $|P|=n$ and $|Q|=O(n)$, determine whether there is a translation of $P$ that makes it contained in $Q$. 

\smallskip
All these problems are known to be 3SUM-hard. If $q=0$, then the \textsc{$q$-Area Triangle} problem is the same as determining whether three points of $S$ are collinear, which is known to be 3SUM-hard~\cite{DBLP:journals/comgeo/GajentaanO12}. If we draw unit disks centered at the points of $S$, then the deepest region in this disk arrangement corresponds to a 
 location for the center of the unit disk that would contain most points. Determining the deepest region in a disk arrangement\footnote{Although the reduction of~\cite{DBLP:journals/siamcomp/AronovH08} uses disks of various radii, it is straightforward to modify the proof with same size disks.} is known to be 3SUM-hard~\cite{DBLP:journals/siamcomp/AronovH08}, which can be used to show the 3-SUM-hardness of \textsc{$q$-Points in a Disk}. Barequet and Har-Peled~\cite{DBLP:journals/ijcga/BarequetH01} showed that the \textsc{Interval Containment} problem is 3-SUM-hard.

Our techniques generalize to a general  \textsc{Pair Search Problem}: Given a problem $P$ of size $n$, where a solution for $P$ can be defined by a pair of elements in $P$, and a procedure $A$ that can verify whether a given pair corresponds to a solution in $O(n^{1+o(1)})$ classical time, determine a solution pair for $P$. Consequently, we obtain $O(n^{1+o(1)})$-time quantum algorithms for the following problems.

    \smallskip\noindent\textsc{Polygon Cutting}: Given a simple $n$-vertex polygon $P$, an edge $e$ of $P$ and an integer $K>2$, is there a line that intersects $e$ and cuts the polygon into exactly $K$ pieces?
    
    \smallskip\noindent\textsc{Disjoint projections}: Given a set $S$ of $n$ convex objects, determine a line such that the set objects project disjointly on that line.

 \smallskip\noindent
The \textsc{Polygon Cutting} problem is known to be 3SUM-hard~\cite{ruci2008cutting}.  \textsc{Disjoint projections} can be solved in $O(n^2 \log n)$ time in classical computing model~\cite{edelsbrunner1983graphics}, but 
 it is not yet known to be 3SUM-hard.

We also show how the pair search can be further generalized for $d$-tuple search or in $\mathbb{R}^d$. 

\section{Preliminaries}

\todo{Consider adding a description of the quantum model we use as done by Ambainis and Larka}
In this section, we describe some standard quantum procedures and tools from the literature that we will utilize to design our algorithms.
\begin{theorem}[Grover Search~\cite{grover1996fast}] Let $X = $ $\{x_{1}, $ $x_{2},$ $...,$ $x_{n}\}$ be a set of $n$ elements and let  $f: X \xrightarrow{} \{0, 1\}$ a boolean function. There is a bounded-error quantum procedure that can find an element $x \in X$ such that $f(x) = 1$ using $O(\sqrt{n})$ quantum queries.
\end{theorem}

\begin{theorem}[Amplitude Amplification~\cite{brassard2002quantum}] Let $A$ be a quantum procedure with a one-sided error and success probability of at least $\epsilon$. Then there is a quantum procedure $B$ that solves the same problem with a success probability $\frac{2}{3}$ invoking $A$ for $O(\frac{1}{\sqrt{\epsilon}})$ times.
\end{theorem}

By repeating Amplitude Amplification a constant number of times we can achieve a success probability of $1-\epsilon$ for any $\epsilon > 0$. This technique has been widely used in the literature to speed up classical algorithms.

\begin{algorithm}
\caption{Recursive-Quantum-Search (RQS)}
\begin{algorithmic}[1] 
\State\textbf{Procedure} RQS($M,n, \delta,\epsilon$), where   
    $M$ is a problem of size at most $n$,  $\delta$ is a positive constant, and $\epsilon$ is an allowable error parameter.
    \If{  $|M|<k$, where  \(k = n^{1/\alpha} \cdot\delta (\log  n  + \log \epsilon^{-1})\) and $\alpha \in  O(\sqrt{ \log n/ \log\log n})$, }
        \State Search for a solution by exhaustive search
    \Else
        \State Let $R_1,\ldots, R_t$ be a decomposition  of the   problem $M$ into $t$   subproblems, where  $t\in O(k^2)$.
         \indent \indent Search  for a solution  that spans two or more  subproblems. 
        \If{any of the subproblems is  larger 
          than $\frac{|M|}{k}\cdot \delta (\log  |M|  + \log \epsilon^{-1})$}
        \State return Error
        \Else
        \State Let $A$ be an algorithm that runs     RQS   ($R, n,\delta,\epsilon$)  recursively on  randomly chosen   \indent \indent subproblem $R$. Run $A$ with Amplitude Amplification  for a success probability at   least   \indent \indent \(1-\epsilon\).
        \EndIf
    \EndIf
\end{algorithmic}
\end{algorithm}

Algorithm 1 presents the Recursive-Quantum-Search (RQS) of Ambainis and Larka~\cite{ambainis2020quantum} for searching over a subdivision, but we slightly modify the description to present it in terms of subproblems. We first describe the idea and summarize it in a theorem (Theorem~\ref{gen}) so that it can be used as a black box. We then illustrate the concept using the Point-On-3-Lines problem.

The algorithm decomposes the problem into $O(k^2)$ subproblems, where $k$ is a carefully chosen parameter, and then checks whether there is a solution that spans at least two subproblems but does not evaluate the subproblems. If no such solution exists,  then the solution is determined by one of the subproblems. If all the subproblems are sufficiently small, then it searches for a solution over them using Grover search; otherwise, it returns an error. Consequently, one needs to show that the probability of a subproblem being large can be bounded by an allowable error parameter $\epsilon$, 
and hence, Grover search will ensure a faster running time.  
   
We now have the following theorem, which is inspired by Ambainis and Larka's~\cite{ambainis2020quantum} result, but we include it here for completeness.
 


\begin{theorem}\label{gen}
Let $M$ be a problem of size at most $n$. Assume that for every $k<|M|$, $M$ can be decomposed into $O(k^2)$ subproblems such that $M$ can be solved first by checking for solutions that span at least two subproblems (without evaluating the subproblems), and then, if such a solution is not found, applying a Grover search over these subproblems (when we evaluate the subproblems). 
Furthermore, assume there exists a constant $\delta$ such that the probability for a subproblem to have a size larger than $\frac{|M|}{k}\cdot \delta (\log  |M|  + \log \epsilon^{-1})$ is at most $\epsilon$, where $\epsilon$ is an allowable error probability. 

If we can compute the problem decomposition and check whether there is a solution that spans at least two subproblems in $O(|M|^{1+o(1)}k^2)$ classical time, then RQS can solve $M$ in $O(n^{1+o(1)})$ quantum time.\todo{base case must be solved in poly time} 
\end{theorem}
\begin{proof}
The first time RQS is called,   $M$ is the original problem with size $|M|=n$. 
Since the recursion tree has a branching factor of $O(k^2)$, the number of problems at level $j$ is $C_1k^{2j}$, where $C_1$ is a constant.

We set $k$ to be $n^{1/\alpha} \cdot\delta (\log  n  + \log \epsilon^{-1})$,  where $\alpha \in O(\sqrt{\log n/\log \log n})$. At each recursion, the problem size decreases by a factor of  $n^{-1/\alpha}$, and at $j$th level, a problem has size at most $n^{1-j/\alpha}$. Since the cost of problem decomposition and checking whether a solution spans two or more subproblems is $O(|M|^{1+o(1)}k^2)$, using Grover search, the cost for level $j$ is $\sqrt{C_1k^{2j}}\cdot C_2(n^{1-j/\alpha}n^{2/\alpha}n^{o(1)})$,\todo{j:check carefully - we decided $n^o(1)$ subsumes polylog} 
where $C_2$ is a constant. We sum the cost of all levels to bound   $T(n)$. 
  

\begin{align*}
T(n) &\le C_2 \sum_{j=0}^{\alpha} \sqrt{(C_1k)^{2j}}\left(n^{1-j/\alpha}n^{2/\alpha} n^{o(1)}\right)\\
&=  C_2 n^{1+(2/\alpha)+o(1)}\sum_{j=0}^{\alpha} \left(\frac{C_1k}{n^{1/\alpha}} \right)^{j}\\
&=   C_2 n^{1+(2/\alpha)+o(1)}\sum_{j=0}^{\alpha} \left(\frac{C_1n^{1/\alpha} \delta (\log  n  + \log \epsilon^{-1})}{n^{1/\alpha}} \right)^{j}\\
& \leq  C_2 n^{1+(2/\alpha)+o(1)}\sum_{j=0}^{\alpha} (C_3\log  n )^{j}\\
&\leq  C_2 \alpha(C_{3}\log n)^{\alpha}n^{1+(2/\alpha)+o(1)}\\
&=  C_2 \alpha \left(\frac{C_{3}\log  n }{n^{ 2/\alpha^2}} \right)^{\alpha} \left(n^{1+2/\alpha}n^{(2/\alpha)+o(1)} \right)
\end{align*} 
 
If \(\alpha{=}\sqrt{\frac{2\log (n)}{\log (C_{3}) + \log\log (n)}}\), then  \(n^{\frac{2}{\alpha^{2}}} { =} C_{3}\log(n)\). Hence 
    $T(n) =  C_2 \alpha n^{1+\frac{4}{\alpha} +o(1)}  
  = O(n^{1+o(1)})$.
  
\end{proof}

 \begin{figure*}[pt]
     \centering
     \includegraphics[width=\textwidth]{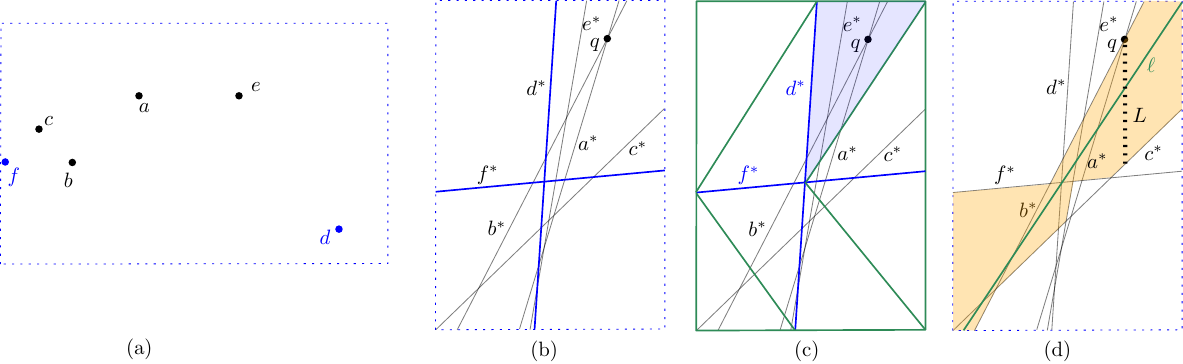}
     \caption{(a)--(b) Illustration for a point set and its corresponding lines in the dual plane. (c) Illustration for $T_k$, where $k=2$, with a face $R$ shown in blue shaded region. The dual  and supporting edges are in blue and green, respectively. (d) Illustration for the zone of a supporting line.}
     \label{fig:boundary}
 \end{figure*}

We can use Theorem~\ref{gen} as a black box. 
For example, consider the case of Point-On-3-Lines problem. Let $S$ be the set of input  lines. To construct subproblems, choose $k$ lines randomly, then create an arrangement of these lines, and finally, triangulate the arrangement to obtain $O(k^2)$ faces. 
{\color{black}Specifically, a subproblem corresponding to a closed face $F$ consists of the input lines that bound $F$ and the lines that intersect the interior of $F$}. If a solution point (i.e., a common point on three lines) spans at least two {\color{black}closed} faces, then it must lie on an edge or coincide with a vertex of the triangulation, which can be checked in $O(|S|k^2 \log n)$ time without evaluating the subproblems. If a solution point is not found, then we can search over the subproblems using Grover search. Ambainis and Larka's~\cite{ambainis2020quantum} showed that there exists a constant $\delta$ such that the probability of a subproblem to contain more than $\delta \frac{|S|}{k}(\log (|S|) + \log (\epsilon^{-1}))$ lines is bounded by $\epsilon$, and hence, we can apply Theorem~\ref{gen}. 
The following lemma, which is adapted from~\cite{ambainis2020quantum}, will be helpful for us to argue about subproblem sizes.
\begin{lemma}[Ambainis and Larka~\cite{ambainis2020quantum}]\label{subroutine}
Let $S$ be a set of straight lines and let $A$ be an arrangement of $k$ lines that are randomly chosen from $S$. Let  $\mathcal{T}$ be a planar subdivision of size $O(k^2)$ obtained by adding straight line segments to $A$ such that each face of $\mathcal{T}$ is of size $O(1)$. 
Then the probability of a closed face of $T$, without its vertices,  intersecting more than $\delta \frac{|S|}{k}(\log (|S|) + \log (\epsilon^{-1}))$ lines of $S$ is bounded by $\epsilon$, where $\delta$ is a positive constant and $\epsilon$ is an allowable error probability. 
\end{lemma}

{\color{black}The reason to restrict the attention to a closed face without its vertices (in Lemma~\ref{subroutine}) is to avoid the degenerate case with many lines intersecting at a common point. In such a scenario,  a random sample of $S$  is likely to have many lines passing through such a point, yielding a closed face intersected by many lines. Ambainis and Larka's proof~\cite{ambainis2020quantum} did not explicitly discuss this scenario.} 

\section{Finding a Triangle of Area at most $q$}

A well-known approach for finding a minimum area triangle among a set $S$ of $n$ points~\cite{DBLP:journals/jcss/EdelsbrunnerG89} is to use point-line duality. For each point $p=(p_{x}, p_{y})$, construct a line  $p^{*}$, which is defined as $y = p_{x}x - p_{y}$ in the dual plane. Figure~\ref{fig:boundary}(b) illustrates a set of lines corresponding to the points of Figure~\ref{fig:boundary}(a). The algorithm uses the property that the line with the smallest vertical distance from the intersection point of a pair of lines $a^*$ and $b^*$ determines a triangle $\Delta abc$ that minimizes the area over all the triangles that must include $a$ and $b$. Consequently, one can first construct a line arrangement in the dual plane and then examine its faces to find a minimum area triangle in $O(n^2)$ classical time.


We now describe our approach using quantum computing. One can check whether three lines in the dual plane intersect at a common point in $O(n^{1+o(1)})$ quantum time~\cite{ambainis2020quantum}, and if so, it would correspond to a triangle of 0 area. Therefore, we may assume that the lines are in a general position.

We first discuss the concept of `zone' in an arrangement and some properties of a minimum area triangle. Let 
$A_k$ be an arrangement of a set $S^*_k$ of $k$ randomly chosen dual lines, and let  $T_k$ be a triangulation obtained from $A_k$ in $O(k^2)$ time. The \emph{zone} of a line $\ell$ is the set of closed faces in $A_k$ intersected by $\ell$.\todo{If a face is only intersected by l at a vertex, is it included in the zone? j: added 'closed'} 
Figure~\ref{fig:boundary}(b)--(c) illustrates a scenario where two lines $d^*,f^*$ have been chosen to create $T_k$. 
For each face $R$ in $T_k$, $s(R)$ denotes the dual lines (a subset of $S^*_k$) that bound $R$ and the dual lines that intersect the interior of $R$. 
We refer to an edge of $T_k$ as a \emph{dual edge} if it corresponds to a dual line of a point in $S$, otherwise, we call it a \emph{supporting edge}. The line determined by a supporting edge is called a \emph{supporting line}. We now have the following property of a minimum area triangle.

\begin{lemma}\label{candidate}
Let $\Delta abc$ be a minimum area triangle. Let $q$ be the intersection point of $a^*$ and $b^*$. Assume that $q$ is not a vertex of $T_k$ and $q$ lies interior to a face $R$ of $T_k$ (e.g.,   Figure~\ref{fig:boundary}(c)). Then either one of the following or both hold: (a) $c^*$ belongs to $s(R)$. 
(b) $c^*$ belongs to $s(Z)$, where $Z$ is a zone of a supporting line of $R$. \todo{c* could also be one of the lines of R, is this captured in case (a)? j:yes}
\end{lemma}
\begin{proof}
Assume that (a) does not hold. We now show that (b) must be satisfied. 
Consider a vertical line segment $L$ starting from $q$ and ending on $c^*$. Since $c^*$ minimizes the vertical distance from $q$, no other dual edge can intersect $L$ (e.g., Figure~\ref{fig:boundary}(d)). Since $q$ is enclosed by $R$ and since $c^*$ is outside of $R$, there must be a supporting edge $\ell$ on the boundary of $R$ that intersects $L$. If the zone of the corresponding supporting line does not contain $c^*$, then we can find a dual line other than $c^*$ that intersects $L$, which contradicts the optimality of $\Delta abc$. Figure~\ref{fig:boundary}(d) illustrates the zone of $\ell$, which is shaded in orange. 
\end{proof}


We now show how to leverage Theorem~\ref{gen}. Let $R_1,\ldots,R_t$ be the faces of $T_k$. We choose $s(R_1),\ldots,s(R_t)$ as the subproblems. By Lemma~\ref{subroutine}, the probability of a  subproblem being large is bounded by $\epsilon$. In Lemma~\ref{lem:suplem}, we show how in $O(nk^2\log n)$ time, one can check whether there is a triangle $\Delta abc$ of area at most $q$ such that no subproblem contains all three dual lines $a^*,b^*,c^*$.  Consequently, we obtain Theorem~\ref{thm:dual}.

\begin{lemma}\label{lem:suplem}
    A triangle that has an area of at most $q$ and spans at least two subproblems can be computed in $O(nk^2 \log n)$ time. 
\end{lemma} 
\begin{proof}
Each candidate triangle $\Delta abc$ satisfies the property that the intersection point $q$ of two of its dual lines lies in some face $R$ and the third dual line does not intersect $s(R)$. Here the condition (b) of Lemma~\ref{candidate} must hold and it suffices to examine the zone of each supporting line of $R$. We thus check the zones of all the supporting lines of $T_k$ as follows. Specifically, given an arbitrary line, its zone in an arrangement of $n$ lines can be constructed in $O(n\log n)$ time~\cite{wang2022simple}. 
Let $e$ be a supporting line and let $Z_e$ be its zone. We search over all the faces of $Z_e$ to find a (vertex, edge) pair, i.e., $(v,L)$, such that they lie on opposite sides of $e$ and minimize the vertical distance from $v$ to $L$. To process a face $F$ we construct two arrays $L_u$ and $L_b$. Here $L_u$ ($L_b$) is an array obtained by sorting the vertices on the upper (lower) envelope of $F$ using x-coordinates in $O(|F|\log |F|)$ time. Since $F$ is convex, for each vertex $q$ in $L_u$ ($L_b$), we can use $L_b$ ($L_u$) to find the dual line that has the lowest vertical distance from $q$ in $O(\log |F|)$ time. Since the number of edges in a zone is $O(n)$~\cite{chazelle1985power}, the total time required for processing all the faces is at most $O(n\log n)$. For $O(k^2)$ supporting lines, the running time becomes $O(k^2n\log n)$.
\end{proof}

\begin{theorem}
    \label{thm:dual}
Given a set $S$ of $n$ points,  one can determine whether there is a triangle with area at most $q$ in \(O(n^{1+o(1)})\) quantum time.    
\end{theorem}


\section{Finding a Unit Disk with at least $q$ Points}

Let $S$ be a set of $n$ points and consider a set $\mathcal{D}$ of $n$ unit disks, where each disk is centered at a distinct point from  $S$. 
Note that to solve \textsc{$q$-Points in a Disk}, it suffices to check whether there is a point $r$ that hits at least $q$ disks in $\mathcal{D}$. However, searching for $r$ using Theorem~\ref{gen} requires tackling some challenges. First, we need to create a problem decomposition, where the probability of obtaining a large subproblem is bounded by an allowable error probability. This requires creating a subdivision (possibly with curves) where the size of each region (corresponding to a subproblem) is $O(1)$. Second, we need to find a technique to check for solutions that span two or more subproblems. 

Consider an arrangement $A_k$ of $k$ randomly chosen disks from $\mathcal{D}$. We first discuss how the regions of $A_k$ can be further divided to create a subdivision $A^\prime_k$ where the size of each region is $O(1)$. 
\begin{lemma}\label{triang}
    Let $A_k$ be an arrangement of $k$ unit disks. In $O(k^2\log n)$ time, one can create a subdivision of $A^\prime_k$ by adding straight line segments such that each face is of size $O(1)$. 
\end{lemma}
\begin{proof} 
For each disk, we create four \emph{pseudolines} as follows. Consider partitioning the disk into four regions by drawing a vertical and a horizontal line through its center. For each circular arc, we create a pseudoline by extending its endpoints by drawing two rays following the tangent lines, as shown  
in Figure~\ref{triang}(a). However, the resulting subdivision may still contain faces with linear complexity (e.g., the face $F$ in Figure~\ref{triang}(b)). We subdivide each face further by 
extending a horizontal line segment from each vertex. 
The details are included in Appendix~\ref{app:tri}. At the end of the construction, each cell of the subdivision can be described using $O(1)$ arcs or segments. The construction inserts at most $O(k^2)$ straight lines and takes $O(\log n)$ time per addition to complete the process in $O(k^2\log n)$ time. 
\end{proof}

 \begin{figure}[pt]
     \centering
     \includegraphics[width=.5\linewidth]{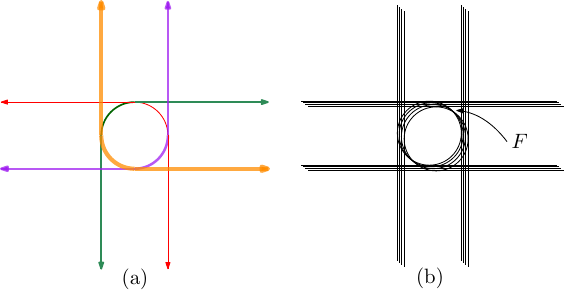}
     \caption{Illustration for the proof of Lemma~\ref{triang}.}
     \label{fig:boundary}
 \end{figure}

We now show how to leverage Theorem~\ref{gen}. Let $R_1,\ldots,R_t$ be the faces of $A^\prime_k$. Let $s(R_i)$, where $1\le i\le t$, be the disks that intersect the closed region $R_i$ (except its vertices), but do not fully contain $R$. We subtract how many disks fully contain $R$ from $q$ and therefore, they should not be considered in the recursive subproblems. 
We show that the probability of a subproblem being large is bounded by $\epsilon$ (Appendix~\ref{app:subroutinePseudo}).  Lemma~\ref{lem:spanning} shows how to check whether there is a solution point $r$ (i.e., a point hitting at least $q$ disks) that coincides with a vertex of $A^\prime_k$ or spans at least two subproblems in $O(nk^2\log n)$ time.  Consequently, we obtain Theorem~\ref{thm:disk}.

\begin{lemma}\label{lem:spanning}
    Let $r$ be a point that hits at least $q$ disks. If $r$ coincides with a vertex of $A^\prime_k$ or spans at least two subproblems then it can be found in $O(nk^2\log n)$ time. 
\end{lemma}
\begin{proof}
For each edge $e=(v,w)$ of $A^\prime_k$, we first count the number of disks intersected by $v$ in $O(n)$ time. We then compute all the intersection points between $e$ and the input disks and sort them based on their distances from $v$ in $O(n\log n)$ time. Finally, we walk along $e$ from $v$ to $w$, and each time we hit an intersection point $o$, we update the current disk count (based on whether we are entering a new disk or exiting a current disk) to compute the number of disks intersected by $o$. 
\end{proof}

\begin{theorem}
    \label{thm:disk}
Given a set $S$ of $n$ points,  one can determine whether there is a unit disk with at least $q$ points in \(O(n^{1+o(1)})\) quantum time.    
\end{theorem}
\section{Determining Interval Containment}

Let $\mathcal{I}$ be an instace of \textsc{Interval Containment}, and let $P=(p_1,\ldots,p_n)$ and $Q=(q_1,\ldots,q_m),$ where $m=O(n)$, be the two sets of pairwise disjoint intervals of $\mathcal{I}$. We now give an  $O(n^{1+o(1)})$-time quantum algorithm to determine whether $P$ can be translated so that it becomes contained in $Q$. If there is an affirmative solution, then we can continuously move the intervals in $P$ until an endpoint of one of its intervals hits an endpoint of an interval of $Q$, as shown in  Figure~\ref{fig:intervals}(c)--(d). 

\begin{remark}\label{rem:point}
If  $\mathcal{I}$ admits an affirmative solution, then there is a solution where an end point of one interval of $P$ coincides with an end point of an interval in $Q$. 
\end{remark}

We now use Remark~\ref{rem:point} to find a solution for $\mathcal{I}$.
 
\begin{figure}[pt]
    \centering   \includegraphics[width=.65\linewidth]{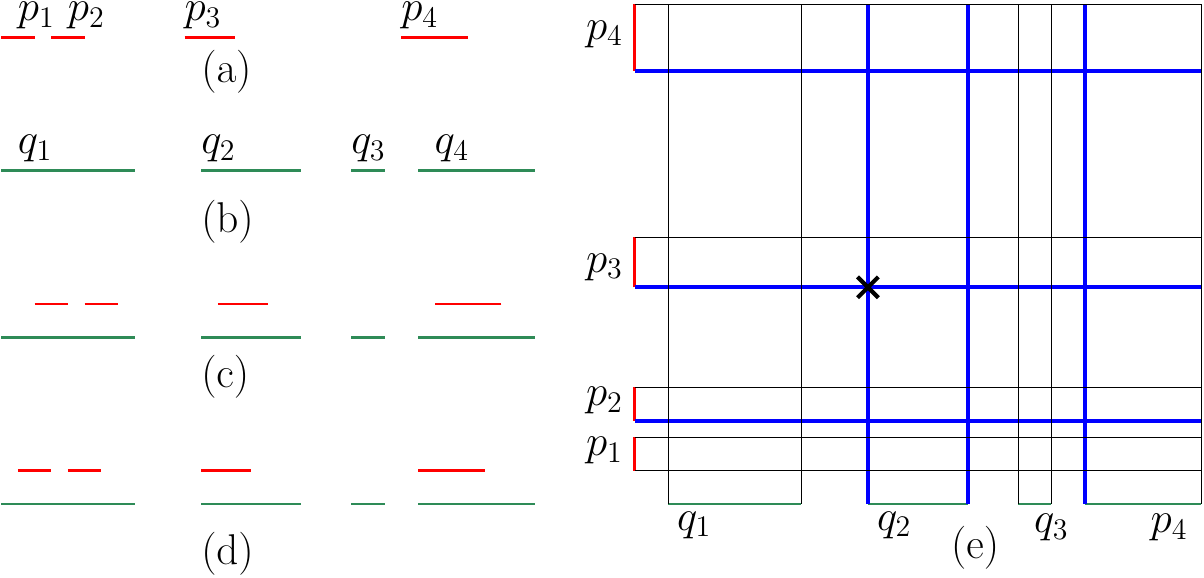}
    \caption{(a) $P$. (b) $Q$. (c) $P\subset Q$.  Illustration for (d) Remark~\ref{rem:point} and  (e) Theorem~\ref{thm:intervals}.}
    \label{fig:intervals}
\end{figure}

\begin{theorem}
    \label{thm:intervals}
Given two sets $P$ and $Q$ of $O(n)$ pairwise-disjoint intervals on a line, one can determine whether there is a translation of $P$ that makes it contained in $Q$ in \(O(n^{1+o(1)})\) quantum time.    
\end{theorem}
\begin{proof}
We place the intervals of $Q$ on the positive x-axis starting from $(1,0)$ and the intervals of $P$ on the positive y-axis starting from $(0,1)$, as shown in Figures~\ref{fig:intervals}(e). Consider a set $H$ of $2n$ horizontal lines and a set $V$ of $2m$ vertical lines through the endpoints of the intervals of $P$ and $Q$, respectively. Let $A_k$ be an arrangement determined by $k$ randomly chosen lines from  $(H\cup V)$, 
e.g., the thick blue lines of Figures~\ref{fig:intervals}(e). Add the smallest area rectangle containing $P$ and $Q$ to the arrangement so that we get a subdivision $T_k$, where its faces $R_1,\ldots, R_t$ are rectangles. By $s(R_i)$, where $1\le i\le t$, we denote the lines of $H$ and $V$ that intersect the closed region $R_i$. 

We now show how to leverage Theorem~\ref{gen}. By Remark~\ref{rem:point}, it suffices to examine pairs of endpoints from $P$ and $Q$.  Let $a$ be an endpoint from $P$ and $b$ be an endpoint from $Q$ that determine a solution. Let $o$ be the intersection point of the corresponding lines $\ell_a\in V$ and $\ell_b\in H$. We refer to $o$ as the solution point, which may lie at a vertex, or interior to an edge, or interior to a face of $T$. 

We choose $s(R_1),\ldots,s(R_t)$ as the subproblems. By Lemma~\ref{subroutine}, the probability of a subproblem being large is bounded by $\epsilon$.  In $O(nk^2 \log n)$ time, we can check whether $o$ coincides with a vertex of $T_k$, i.e., spans at least two subproblems (Figure~\ref{fig:intervals}(e)). However, we do not check whether the solution $o$ lies on an edge of $T$ because if $o$ lies interior to an edge or a face of $T_k$, then it is found by a Grover search over the subproblems. The running time follows directly from Theorem~\ref{gen}. 
\end{proof}

\section{Pair/Tuple Search and Generalizations}

For a pair search problem $P$, if we can decide whether a given pair corresponds to a solution in $f(n)\in O(n^{o(1)})$ classical time, then a straightforward application to Grover search yields an $O(n^{1+o(1)})$-time quantum algorithm. However, we show how to solve $P$ in $O(n^{1+o(1)})$-time even when $f(n)\in O(n^{1+o(1)})$.

\begin{theorem}\label{thm:pairmagic}
    Let $P$ be a problem of size $n$ where a solution for  $P$ can be defined by a pair of elements in $P$. Assume that we can decide whether a given pair corresponds to a solution in $O(n^{1+o(1)})$ classical time. Then a solution pair can be computed in $O(n^{1+o(1)})$ time using a quantum algorithm.
\end{theorem}
\begin{proof}
    We first label the elements of $P$ from $t_1$ to $t_n$. For each element $t_i$, we create a horizontal line $y=i$ and a vertical line $x=i$. Every pair of lines $(a,b)$, where one is horizontal and the other is vertical, 
    corresponds to a pair of elements $(t_a,t_b)$. 
    Now the search over the subdivision is similar to the proof of Theorem~\ref{thm:intervals}.
\end{proof}

Theorem~\ref{thm:pairmagic} allows for an $O(n^{1+o(1)})$-time quantum algorithm for \textsc{Polygon Cutting} and \textsc{Disjoint projections} problems. The details are in Appendix~\ref{app:pc}. The pair search technique can be applied to obtain quantum speed up as long as the check for a pair takes sub-quadratic time. For example, if a pair can be checked in $O(n^{1+\beta})$ classical time, then the analysis of Theorem~\ref{gen} gives an algorithm with $O(n^{1+\beta})$ quantum time. Hence a maximum clique in a unit disk graph, where pairs of points are checked in $O(n^{1.5}\log n)$ classical time~\cite{DBLP:conf/wg/Eppstein09,DBLP:conf/compgeom/EspenantKM23},  can be found in $O(n^{1.5})$ quantum time. The pair search technique generalizes to $d$-tuple search, where one needs to search for a solution over an arrangement in $\mathbb{R}^d$. Appendix~\ref{app:ts} includes the details. 


\begin{theorem}\label{thm:pairmagic3d}
    Let $P$ be a problem of size $n$ where a solution for  $P$ can be defined by a $d$-tuple of elements in $P$, where $d\in O(1)$. Assume that we can decide whether a given tuple corresponds to a solution in $O(n^\beta)$ classical time, where $\beta>0$. Then a solution for $P$ can be computed in $O(n^{1+\beta+o(1)})$ time using a quantum algorithm.
\end{theorem}

\section{Conclusion}
In this paper we discuss quantum speed-up for some geometric 3SUM-Hard problems. We also show how our technique can be applied to a more general pair or tuple search setting. A natural avenue to explore would be to establish nontrivial lower bounds under quantum 3SUM conjecture. 


\small
\bibliographystyle{abbrv}
\bibliography{refs} 
\clearpage
\newpage
\normalsize
\appendix
\section{Subdividing an Arrangement of Straight Line Segments and Circular Acrs}
\label{app:tri}

In this section, we show how to subdivide an arrangement of $O(k^2)$ line segments and circular arcs in $O(k^2\log k)$ time such that the region corresponding to each cell can be described with $O(1)$ straight line segments and circular arcs. Note that it suffices to subdivide each face of the arrangement independently, because a subproblem is determined by the subset of disks that intersect a cell, which is independent of any information regarding other cells in the arrangement. Specifically, let $F_1$ and $F_2$ be two faces in the arrangement as shown in Figure~\ref{fig:sweep}(a), which are subdivided with straight line segments (shown in gray). Since they are subdivided independently, the cell $C$ of $F_2$ does not contain the information about the vertices that are added on its boundary when $F_1$ is being subdivided. In other words, the cell $C$ can be described only by the two line segments that appear inside $F_2$ and an arc that appears on $F_2$. 

In the following, we show how to subdivide an $m$-vertex face $P$ in $O(m\log m)$ time such that the region corresponding to each cell can be described with $O(1)$ straight line segments and circular arcs. 

\begin{figure*}[pt]
    \centering
    \includegraphics[width=\linewidth]{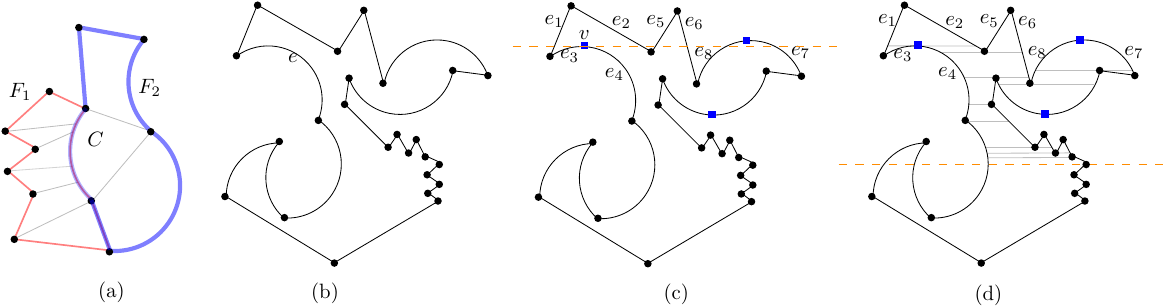}
    \caption{(a) Illustration for $F_1$ and $F_2$ in red and blue, respectively. (b) A face $F$. (c)--(d) Sweeping a line (shown in dashed orange) to subdivide $P$. }
    \label{fig:sweep}
\end{figure*}

Since a straight line can be considered as an arc of a circle with an infinite radius, in the following we do not distinguish between segments and arcs. We first subdivide each arc of $P$ by adding (at most two) dummy vertices on the arc such that in the resulting face $P^\prime$, all arcs become y-monotone. For example, the edge $e$ in Figure~\ref{fig:sweep}(a) is split into $e_3$ and $e_4$, as shown in Figure~\ref{fig:sweep}(b). Since we need to add at most $O(m)$ dummy vertices, this step takes $O(m)$ time. 

We sort the vertices of $P^\prime$ in decreasing order of their $y$-coordinates. We then sweep $P$ downward with a horizontal line $\ell$ from the topmost vertex of $P$. While sweeping the plane, we keep a dynamic binary search tree $\mathcal{T}$ over the arcs that are being intersected by $\ell$~\cite{mark2008computational}.   $\mathcal{T}$ is updated to process the events, i.e., when we reach the bottommost or topmost endpoint of an edge, where an update operation on $\mathcal{T}$ takes $O(\log m)$ time. 

Each time the sweep line hits a vertex $v$, we examine its incident edges. We use $\mathcal{T}$ to find the arcs $e_a$ and $e_b$ to the left and right of $v$, respectively, which are horizontally visible to $v$ inside $P^\prime$. For example, in Figure~\ref{fig:sweep}(b),  the left and right edges of $v$ are $e_1$ and $e_2$, respectively. Note that sometimes $e_a$ and $e_b$ may correspond to edges that are incident to $v$ (e.g., consider the topmost vertex of $P^\prime$).  We examine $e_a$ and $e_b$ to determine whether the horizontal visibility lines from $v$ to these edges lie inside the face, and if so, we draw a line segment through $v$ to hit these edges. 

Every time a horizontal visibility line is drawn, we obtain a new cell above this line in the subdivision. A cell cannot have any vertex on its left and right sides; otherwise, that vertex would split the cell further. Therefore, a cell can be described by two arcs bounding the left and right sides and at most two line segments bounding the top and bottom sides. Since we have $O(m)$ update operations,  the overall time complexity is $O(m\log m)$.

\section{Computing Subproblem Size with Pseudolines}
\label{app:subroutinePseudo}

The following lemma is a generalization of Lemma~\ref{subroutine}.

\begin{lemma}\label{subroutinePseudo}
Let $S$ be a set of simple pseudolines in $\mathbb{R}^2$, where every pair of lines intersect at most $O(1)$ times, and let $A$ be an arrangement of $k$ pseudolines that are randomly chosen from $S$. Let  $\mathcal{T}$ be a planar subdivision of size $O(k^2)$ obtained by adding some straight line segments to $A$ such that each face of $\mathcal{T}$ is of size $O(1)$. Then the probability of a closed face of $\mathcal{T}$, without its vertices,  intersecting more than $\delta \frac{|S|}{k}(\log (|S|) + \log (\epsilon^{-1}))$ pseudolines of $S$ is bounded by $\epsilon$, where    $\delta$ is a positive constant and  $\epsilon$ is an allowable error probability. 
\end{lemma}
\begin{proof}
We adapt the proof of Ambainis and Larka~\cite{ambainis2020quantum}. Let $P$ be the chosen pseudolines. Assume that no pseudolines intersect more than $C$ times, where $C$ is a constant, and let $l$ be any straight line segment or pseudoline of $A$. Let $X_{1}, X_{2},...,X_{m}$ be the ordered intersections with pseudolines from $S$, where $m \leq C|S|$. We color the intersection points corresponding to the pseudolines of $P$ white, otherwise, color them black. Define $L = \frac{C|S|}{k}(5\log(C|S|) + \log(\epsilon^{-1}))$. We say a pseudo-line $l$ is bad if it has $L$ consecutive black intersections. 
\begin{align*}
     Pr[l \text{ is bad} ] &= Pr[\bigvee_{i=1}^{m-L+1} (X_{i}, X_{i+1},..., X_{i+L-1} \text{ are black})]\\
    &\leq Pr[\sum_{i=1}^{m-L+1} (X_{i}, X_{i+1},..., X_{i+L-1} \text{ are  black})]\\
    &\le  (m-L+1)\frac{{C|S|-L\choose |W|}}{{C|S|\choose |W|}}, \text{where $|W|$ is the} \\ &\text{number of white intersection points}\\
    &\leq C|S|\frac{ (C|S|-L)^{|W|}}{(C|S|)^{|W|}}
\end{align*}
From Ambainis and Larka~\cite{ambainis2020quantum}, this is bounded by $\frac{\epsilon}{C^{4}|S|^{4}}$. Since the size of $\mathcal{T}$ is $O(k^2)\in O(|S|^4)$, the probability of $\mathcal{T}$ to contain a bad pseudoline is bounded by $\epsilon$. Let $F$ be a face of $\mathcal{T}$ with $\beta$ edges where $\beta$ is a constant. Since $F$ does not contain any bad line, the number of pseudolines that can intersect $F$ is at most $\beta(L-1)$, as required. 
\end{proof}

\section{Polygon Cutting and Disjoint Projections}\label{app:pc}
Let $\mathcal{I}$ be an instance of \textsc{Polygon Cutting}.  If $\mathcal{I}$ has an affirmative solution $\ell$ that intersects the given edge $e$ and cuts the input polygon $P$  into $K$ pieces, then a solution can be described by a pair of vertices of $P$ as follows: First move $\ell$ continuously until it hits a vertex of $P$ and then continuously rotate $\ell$ clockwise anchored at $v$ until it hits another vertex $v$.

If the pair $(u,v)$ is specified, then a corresponding solution (if exists) can be computed as follows. Let $\ell_1$ be the line through $u,v$. Sweep $\ell_1$ to find a line $\ell_2$ parallel to $\ell_1$ that lies on one half-plane of $\ell_1$ such that no vertex of $P$ appears between $\ell_1$ and $\ell_2$. Similarly, find another line $\ell_3$ on the other half-plane of $\ell_1$. Considering the construction of $\ell$, one of the lines among  $\ell_2$ and $\ell_3$ corresponds to a solution. Given a pair $(u,v)$, such a check can be done in $O(n\log n)$ time. By Theorem~\ref{thm:pairmagic}, we obtain a solution for \textsc{Polygon Cutting} in $O(n^{1+o(1)})$ quantum time.

We can apply a similar technique for the \textsc{Disjoint Projections} problem, and thus obtain the following theorem.

\begin{theorem}
  An instance $\mathcal{I}$ of \textsc{Polygon Cutting} (or, \textsc{Disjoint Projections}) of size $n$ can be solved in $O(n^{1+o(1)})$ quantum time, where $n$ denotes the size of $\mathcal{I}$.
\end{theorem}

 

\section{Tuple Search}\label{app:ts}
One can also generalize the pair search to $k$-tuple search. In this case, we can arrange the $n$ elements of $S$ on $d$ basis vectors that are pairwise orthogonal, and search over a $d$-dimensional arrangement. 

Note that we have a set of $S^\prime$ hyperplanes, where $|S^\prime|=d|S|$, and it can be partitioned into $d$ subsets $S^\prime_1,\ldots,S^\prime_d$ where the hyperplanes in each subset are parallel to each other. Let $H$ be a set of $k$ randomly chosen hyperplanes from $S^\prime$  and let $\mathcal{A}$ be the arrangement of $H$. Every cell $C$ of $\mathcal{A}$ is a hypercube and thus contains $2d$ faces of dimension $(d-1)$. 

We now show that the subproblem sizes would be small, i.e., no cell is intersected by $L$ hyperplanes from $(S^\prime\setminus H)$, where $L = \frac{2d^2|S^\prime|}{k}(5\log(|S^\prime|) + \log(\epsilon^{-1}))$.  Let $f_j$, where $1\le j\le 2d$, be a face of cell $C$ with its normal parallel to the $j$th basis vector. If a cell $C$ is intersected by at least $L$ hyperplanes from $(S^\prime\setminus H)$, then a face $f_j$ is intersected by at least $\frac{L}{2d}$ of these hyperplanes, and $\frac{L}{2d^2}$ of these belong to some $S^\prime_i$.  We now show that the probability of a face $f_j$ being intersected by at least $\frac{L}{2d^2}$  hyperplanes from $(S^\prime_i\setminus H)$ is bounded by $\epsilon$.

Let $h$ be such a hyperplane determined by $f_j$.  Since the hyperplanes of  $S^\prime_i$ are parallel to each other and since each hyperplane of $S^\prime_i$ is perpendicular to $h$, we can order these hyperplanes based on their distances from the origin. We now examine the probability of $h$ being bad, i.e., having a face intersected by at least $L^\prime = \frac{L}{2d^2}$ hyperplanes of $S^\prime_i$. We now apply the same analysis as we did in the proof of Lemma~\ref{subroutinePseudo}, as follows. We can set aside $L^\prime$ consecutive hyperplanes that intersect a cell in $(|S^\prime_i|-L^\prime-1)$ ways, and for each option, $H$ can be chosen from the remaining hyperplanes in $\frac{{(|S^\prime_i|-L^\prime)\choose |H|}}{ {(|S^\prime_i|)\choose |H|}}$ ways. Here the term $(|S^\prime_i|-L^\prime-1) \frac{{(|S^\prime_i|-L^\prime)\choose |H|}}{ {(|S^\prime_i|)\choose |H|}} \le |S^\prime_i| (\frac{(|S^\prime_i|-L^\prime)^ |H|} {(|S^\prime_i|)}) ^{|H|}$ is bounded by  $\frac{\epsilon}{|S^\prime_i|^4}$~\cite{ambainis2020quantum}. Note that for every $S^\prime_i$, $|S^\prime_i| = |S|$. Therefore, the probability of $h$ to be bad over all $S^\prime_i$, where $1\le i\le d$, is $\frac{d\epsilon}{|S|^4}$, and the probability of any of the $k$ hyperplanes of $\mathcal{A}$ to be bad is $\frac{|H|d\epsilon}{|S|^4}\le \epsilon$.


We now extend the proof of Theorem~\ref{gen} to $d$-dimensions by setting $k = 2d^2n^{1/\alpha} \cdot\delta (\log  n  + \log \epsilon^{-1})$, where the algorithm will return an error if any of the subproblems is larger than $L$. Given a $d$-tuple, if one can check whether it corresponds to a solution in $O(n^{\beta})$ time, where $\beta\ge 0$, then we 
obtain the following time complexity with an analysis similar to Theorem~\ref{gen}. 
The only exception is that the branching factor and construction time of the arrangement both increase to $O(k^d)$~\cite{halperin2017arrangements}, where $k=|H|$.

\begin{align*}
T(n) &\le C_2 \sum_{j=0}^{\alpha} \sqrt{(C_1k)^{dj}}\left(n^{(1-j/\alpha)}n^\beta n^{d/\alpha} n^{o(1)}\right)
\end{align*} 
\begin{align*} 
&\le  C_2 n^{1+\beta+(d/\alpha)+o(1)}\sum_{j=0}^{\alpha} \left(\frac{C_1k^{d/2}}{n^{1/\alpha}} \right)^{j}\\
&=   C_2 n^{1+\beta+(d/\alpha)+o(1)}\sum_{j=0}^{\alpha} \left(\frac{C_3n^{d/2\alpha} \delta (\log  n  + \log \epsilon^{-1})}{n^{1/\alpha}} \right)^{j}, \text{when $d\in O(1)$}\\
& \leq  C_2 n^{1+\beta+(d/\alpha)+o(1)} n^{(d-2)/2\alpha} \sum_{j=0}^{\alpha} (C_4  \log  n )^{j}\\
&\leq  C_2 \alpha(C_{4}   \log n)^{\alpha}n^{1+\beta+(d/\alpha)+({(d-2)/2\alpha})+o(1)}\\
&=  C_2 \alpha \left(\frac{C_{4}\log  n }{n^{ d/\alpha^2}} \right)^{\alpha} \left(n^{1+\beta+(2d/\alpha)+({(d-2)/2\alpha})+o(1)} \right)
\end{align*} 

 
If \(\alpha{=}\sqrt{\frac{d\log (n)}{\log (C_{4}) + \log\log (n)}}\), then  \(n^{\frac{d}{\alpha^{2}}} { =} C_{4}\log(n)\). Hence 
    $T(n) =  C_2 \alpha n^{1+\beta+\frac{5d-2}{2\alpha} +o(1)}$. For $d\in O(1)$, this is bounded by $O(n^{1+\beta+o(1)})$.

\end{document}